\documentclass[11pt]{article}
\usepackage{epsfig}
\usepackage{amssymb,amsmath,amsthm,amscd}
\usepackage{latexsym}
\usepackage{color}

\usepackage{diagbox}
\usepackage{multirow}
\usepackage{makecell}

 \newtheorem{thm}{Theorem}[section]
 \newtheorem{lem}[thm]{Lemma}
  \newtheorem{rem}[thm]{Remark}

 \newtheorem{prop}[thm]{Proposition}
 \newtheorem{cor}[thm]{Corollary}
 \theoremstyle{definition}

\newcommand{\R}{\mathbb {R}}
\newcommand{\C}{\mathbb {C}}

\newcommand {\mat}      [1] {\left[\begin{array}{#1}}
\newcommand {\rix}          {\end{array}\right]}
\newcommand {\de}      [1] {\left|\begin{array}{#1}}
\newcommand {\nt}          {\end{array}\right|}

\begin{document}
\begin{titlepage}
 \title{Flux ratios for  effects of permanent charges  on ionic flows with three ion species: Case study (II)}
\author{Ning Sun\footnote{College of Mathematics, Jilin University, 2699 Qianjin Street, Changchun, Jilin 130012, P. R. China ({\tt sunning16@mails.jlu.edu.cn}).}\; and  Weishi Liu\footnote{Department of Mathematics, University of Kansas, 
1460 Jayhawk Blvd., Room 405,
Lawrence, Kansas 66045, USA ({\tt wsliu@ku.edu}).  
}}  
\date{}
\end{titlepage}
 
\maketitle

\begin{abstract}
In this paper, we  study effects of permanent charges on ion flows through membrane channels
via a quasi-one-dimensional classical Poisson-Nernst-Planck system. 
 This system includes three ion species, two cations with different valences and one anion,
and permanent charges with a simple structure, 
zeros at the two end regions and a constant over the middle region.
For small  permanent charges,
our main goal is to analyze the effects of permanent charges on ionic flows,
interacting with the boundary conditions and channel structure. Continuing from a previous work, we investigate the problem for a new case toward a more comprehensive understanding about effects of permanent charges on ionic fluxes.

 \end{abstract}
 
\section{Introduction}
\setcounter{equation}{0}
Electrodiffusion  exhibits rich phenomena and plays a central role for many applications
(\cite{Bez00, Eis90, Eis03, Eis12, IBR, IR1,KBA}).
Ionic flow through ion channels is one of critical topics of physiology. 
Ion channels are large proteins embedded in cell membranes 
that provide pathways for electrodiffusion of ions
(mainly Na$^+$, K$^+$, Ca$^{2+}$ and Cl$^-$) between inside and outside of cells
(\cite{Hille, H51, HH52a, HHK49, HK49,SN95, Ussing2}).
Thus, ion channels permit permeation and selectivity,
and produce electric signals for cells to communicate with each other.

Ion channels are defined by their structural characteristics,
channel shapes and permanent charge distributions,
which are responsible for biological functions of ion channels.
The shape of a typical channel could be approximated as a cylindrical-like domain,
with a non-uniform cross-section area.
Within an ion channel, amino acid side chains are distributed mainly over a
relatively “short” and “narrow” portion of the channel, 
with acidic side chains contributing negative charges and basic side chains contributing positive charges.
Permanent charges play a major role for controlling ionic flow properties, interacting 
with other important physical parameters, such as boundary concentrations, 
boundary potential and ion valences of ionic mixtures
(\cite{AEL,  Eisenb03, EHL10,EL07, JEL19,Liu05, Liu09, Rub90, SN09, SGNE,WZCX, ZEL19,ZL20}).

Due to the limitation of present experimental techniques,
the most basic functions of ion channels such as permeation and selectivity
are mainly extracted from the I-V relation measured experimentally (see, e.g., \cite{BCEJ97,BCB05,CE93}).
The I-V relation defines a functional response of the channel structure and boundary conditions (see display (\ref{IVrelation}))
and is an input-output type information of an average effect of the full dynamics of ionic flows.
Individual ionic fluxes carry more information than the I-V relation,
but it is expensive and challenging to measure them (\cite{HK55,JEL19, Ussing}).
A point is that it is still not possible to ``measure/observe" internal dynamical behaviors of ionic flows,
which makes it difficult to understand ion channel properties from experimental data
due to extremely rich phenomena that can be created by the multi-scale feature and the nonlinear interplay among those physical parameters.

Mathematical modelings, analyses and numerical simulations of ion channel problems provide an alternative  and complementary approach for
explaining observed biological phenomena  and discovering new ones.
Poisson-Nernst-Planck (PNP) systems serve as basic primitive models for ionic flows through ion channels (see, e.g.,\cite{Bar92, BCE92,EHL10,GNE,HEL10,IBR,IR1,KBA,LE20, LXE17, Rub90,WZCX}).
There have been a great deal successes in analyzing  PNP models (see, e.g., \cite{AEL, Bar92, BCE92, BCEJ97,  ELX15, HLY21,HEL10,  LW10,MEL20, ML20, PJ97, SL22, WHWH14, ZEL19, ZL20}), particularly, those by geometric singular perturbation (GSP) theory (see, e.g.,  \cite{CWZZ21, EL07, ELX15, JEL19, JL12, LTZ12, JLZ15,  LLYZ, Liu05, Liu09, Liu18,LX15, MEL20, ML20,SL18,YXL22,ZEL19, ZL20}),
which makes it possible to explain some effects of different parameters on the physical properties.

For ionic flows involving two ion species (one cation and one anion), effects of permanent charges 
 have been extensively examined and important phenomena, some counterintuitive, were revealed (\cite{HLY21, JLZ15, Liu18, ZEL19,ZL20}). In terms of flux ratios introduced in \cite{Liu18}, major findings for flows of two ion species are as follows. Depending on boundary conditions, {\em a positive permanent charge} can enhance the anion flux while inhibiting that of cation, can enhance the fluxes of both anion and cation,  can inhibit the fluxes of both anion and cation, but, cannot enhance the cation flux while inhibiting that of cation; and, independent of boundary conditions,   {\em  a positive permanent charge} always helps the flux of  anion more than that of cation.

Recently there are several works on ionic flows with three ion species, 
two cation species with the {\em same valence} and one anion species, 
and some interesting results are obtained on competition of the two cations 
(see, e.g., \cite{BWZ21, WZZ21, Zhang21}). These results on competition of the two cations are obtained with measurements that are different from flux ratios. 
However, for three ion species of different valences, the analysis on PNP system is very limited. 
In \cite{SL22}, the authors started  to examine effects of small permanent charges on ionic flows involving three ion species (two cations with different valences and one anion).  They treated a case with equal chemical potential difference and
discovered several new phenomena,  particularly, in sharp contrast to two ion species cases, {\em a positive permanent charge can enhance the flux of {\em either} one of the cation species more than the anion flux.} 
  A $3\times 3$ matrix $D_0$ (see Section \ref{KVtheta} below), determined by the fluxes associated to zero permanent charge, plays a crucial role for the study in \cite{SL22}. The matrix $D_0$ always has a zero eigenvalue. The case studied in \cite{SL22} concerns the situation when the other two eigenvalues $\sigma_{10}$ and $\sigma_{20}$ of $D_0$ are real.     We consider, in this work, the case when $\sigma_{10}$ and $\sigma_{20}$ are a pair of complex conjugate eigenvalues and establish several abstract results -- counterparts of those provided in \cite{SL22}.  In order to draw concrete results, we further  conduct a detailed study when $\sigma_{10}$ and $\sigma_{20}$ are a pair of pure imaginary eigenvalues. 
In addition to some results that are consistent with those for the case studied in \cite{SL22},  our results show that, in this new case,   {\em small positive  permanent charges can enhance {\em only}
 the flux of the cation species with {\em the smaller valence} more than that of the anion}.

The  paper is organized as follows. In the rest of this introduction, we recall the quasi-one-dimensional PNP model for
ionic flows and the flux ratio for permanent charge effects on individual fluxes,
and present the setup of our study.
In Section \ref{reviews}, we briefly review the geometric singular perturbation theory
developed for PNP models in \cite{Liu09, LX15} and some relevant results from \cite{SL22}.
In Section \ref{SS2} contains the new case study on the flux ratios for permanent charge effects.
We conclude the paper with a brief summary in Section \ref{summary}.

\subsection{A quasi-one-dimensional PNP model for ion transports}
For a mixture of $n$ ion species, a quasi-one-dimensional PNP model (\cite{LW10,NE98}) is
\begin{align}
\begin{aligned}\label{E0}
&\frac{1}{A(X)} \frac{{\rm d}}{{\rm d} X}
 \bigg(\varepsilon_{r}(X) \varepsilon_{0} A(X) \frac{{\rm d} \Phi}{{\rm d} X}\bigg)
 =-e_{0}\bigg(\sum_{s=1}^{n}z_{s}C_{s}+\mathcal{Q}(X)\bigg) \\
&\frac{{\rm d} \mathcal{J}_{k}}{{\rm d} X}=0,~~
 -\mathcal{J}_{k}
 =\frac{1}{k_{B} T} \mathcal{D}_{k}(X) A(X) C_{k} 
 \frac{{\rm d} \mu_{k}}{{\rm d} X},~~ 
 k=1,2, \cdots, n,
\end{aligned}
\end{align}
where 
$X \in[0, l]$ is the coordinate along the longitudinal axis of the channel, 
$A(X)$ is the area of cross-section of the channel at the location $X$,
$\varepsilon_{r}(X)$ is the relative dielectric coefficient, $\varepsilon_{0}$ is the vacuum permittivity, 
$e_{0}$ is the elementary charge,  $\mathcal{Q}(X)$ is the permanent charge density, 
$k_{B}$ is the Boltzmann constant,  $T$ is the absolute temperature;
$\Phi$ is the electric potential, and, for the $k$-th ion species, 
$z_{k}$ is the valence (the number of charges per particle),  $C_{k}$ is the concentration, 
$\mathcal{J}_{k}(X)$ is the flux density through the cross-section over $X$, 
$\mathcal{D}_{k}$ is the diffusion coefficient,  and 
$\mu_{k}$ is the electrochemical potential depending on $\Phi$ and $C_{k}$.

Equipped with system $(1.1)$, 
a meaningful boundary condition for ionic flows through ion channels 
is, for $k=1,2, \cdots, n$,
\begin{align}\label{BVC0}
\Phi(0)=\mathcal{V},~~ C_{k}(0)=\mathcal{L}_{k}>0;~~ 
\Phi(l)=0,~~ C_{k}(l)=\mathcal{R}_{k}>0.
\end{align}
We are interested in study the boundary value problem (BVP) (\ref{E0}) and (\ref{BVC0}) for understanding effects of permanent charges ${\cal Q}$ on ionic fluxes ${\cal J}_k$'s.

As in many mathematical analyses of the BVP (\ref{E0}) and (\ref{BVC0}) (see, e.g. \cite{HLY21,Liu18,SL22,ZEL19,ZL20}), 
we will work with  a dimensionless form. 
Let $C_{0}$ be a characteristic concentration of the problems, for example,
$$
C_{0}
=\max_{1 \leq k \leq n}\big\{\mathcal{L}_{k},\mathcal{R}_{k},\sup_{X \in[0,l]}|\mathcal{Q}(X)|\big\}.
$$
Set
$$
\mathcal{D}_{0}
=\max_{1 \leq k \leq n}\big\{\sup_{X\in[0, l]} \mathcal{D}_{k}(X)\big\} 
~\text { and }~
\hat{\varepsilon}_{r}=\sup_{X \in[0, l]} \varepsilon_{r}(X).
$$
Let
\begin{align*}
\begin{aligned} 
& 
x=
\frac{X}{l},~~ 
h(x)
=\frac{A(X)}{l^{2}},~~ 
\bar{\varepsilon}_{r}(x)
=\frac{\varepsilon_{r}(X)}{\hat{\varepsilon}_{r}},~~ 
\varepsilon^{2}
=\frac{\bar{\varepsilon}_{r} \varepsilon_{0} k_{B} T}{e_{0}^{2} l^{2} C_{0}},~~ 
\phi(x)
=\frac{e_{0}}{k_{B} T} \Phi(X),\\ 
& 
c_{k}(x)
=\frac{C_{k}(X)}{C_{0}},~~ 
Q(x)
=\frac{\mathcal{Q}(X)}{C_{0}},~~  
\bar{J}_{k}
=\frac{\mathcal{J}_{k}}{l C_{0} \mathcal{D}_{0}},~~
D_{k}(x)
=\frac{\mathcal{D}_{k}(X)}{\mathcal{D}_{0}},~~
\bar{\mu}_{k}
=\frac{1}{k_{B} T} \mu_{k}.
\end{aligned}
\end{align*}

In terms of the new variables, the BVP (\ref{E0}) and (\ref{BVC0}) becomes
\begin{align}
\begin{aligned}\label{PNPE1}
&
\frac{\varepsilon^{2}}{h(x)}
\frac{{\rm d}}{{\rm d}x}\bigg(\bar{\varepsilon}_{r}(x)h(x) \frac{d\phi}{{\rm d} x}\bigg)
=-\sum_{s=1}^{n}z_{s}c_{s}-Q(x),\\
&
\frac{{\rm d} \bar{J}_{k}}{{\rm d} x}=0,~~
-\bar{J}_{k}
=h(x)D_{k}(x)c_{k}\frac{{\rm d} \bar{\mu}_{k}}{{\rm d} x},~~k=1,2,\ldots,n,
\end{aligned}
\end{align}
with the boundary conditions at $x=0$ and $x=1$
\begin{align}\label{PNPBV1}
\phi(0)=V=\frac{e_{0}}{k_{B} T} \mathcal{V},~~ 
c_{k}(0)=l_{k}=\frac{\mathcal{L}_{k}}{C_{0}};~~ 
\phi(1)=0,~~ 
c_{k}(1)=r_{k}=\frac{\mathcal{R}_{k}}{C_{0}}.
\end{align}

\begin{rem}\label{swapsym} 
One has the following symmetry of the boundary value problem:
If $(\phi(x), c_k(x),\bar{J}_k)$ is a solution of (\ref{PNPE1}) and (\ref{PNPBV1}), then
\[(\phi^*(x), c^*_k(x),\bar{J}^*_k)=(\phi(1-x),c_k(1-x),-\bar{J}_k)\]
is a solution of (\ref{E0}) with the boundary conditions
\[\phi(0)=-V,~~ c_{k}(0)=r_{k};~~ \phi(1)=0,~~ c_{k}(1)=l_{k}.\]
\end{rem}

For boundary conditions, to avoid sharp boundary layers (\cite{ZEL19, ZL20}),
one often designs boundary conditions to meet the electroneutrality condition 
\begin{align}\label{ENC}
\sum_{s=1}^{n}z_{s}l_{s}=\sum_{s=1}^{n}z_{s}r_{s}=0.
\end{align}

The electrochemical potential 
$\bar{\mu}_{k}(x)=\bar{\mu}_{k}^{id}(x)+\bar{\mu}_{k}^{ex}(x)$ 
for the $k$ th ion species consists of 
the ideal component $\bar{\mu}_{k}^{id}(x)$ given by
\begin{align}\label{muid}
\bar{\mu}_{k}^{id}(x)=z_{k} \phi(x)+\ln c_{k}(x),
\end{align}
and the excess component $\bar{\mu}_{k}^{e x}(x)$.
The {\em classical} PNP 
 model only deals with the ideal component $\bar{\mu}_{k}^{id}$, 
reflecting the collision between ion particles 
and water molecules and ignoring the size of ions. 
The excess electrochemical potential $\bar{\mu}_{k}^{ex}$ 
accounts for the finite size effect of ions
(\cite{Bik42, GNE, GNE1, GNE2, HEL10, JL12, Li09,  LLYZ,LTZ12, Ros89, Ros93, SLBE02, SL18,TR97}).

An important quantity for characterizing ion channel properties is 
the so-called I-V (current-voltage) relation defined as follows. 
For fixed $l_{k}$'s and $r_{k}$'s, 
a solution $(\phi, c_{k}, \bar{J}_{k})$ of the BVP (\ref{PNPE1}) and (\ref{PNPBV1}) 
will depend on the voltage $V$ only, 
and the current $\mathcal{I}$, {\em the~flow~rate~of~charges}, 
is thus related to the voltage $V$ given by
\begin{align}\label{IVrelation}
\mathcal{I}=\sum_{s=1}^{n}z_{s}\bar{J}_{s}(V).
\end{align}

\subsection{Flux ratios for permanent charge effects on ionic fluxes}

Recall the concept of flux ratio for permanent charge effects on ionic fluxes introduced in \cite{Liu18}.
  For fixed boundary conditions $(V,L,R)$  where $L=(l_{1}, l_{2}, \ldots, l_{n})^{T}$ and $R=(r_{1}, r_{2}, \ldots, r_{n})^{T}$,
let $\bar{J}_k(Q)$ be the flux of the $k$-th ion species associated with the permanent charge $Q$,
then the {\em flux ratio} for the $k$th ion species is 
\begin{align}\label{fluxratio}
\lambda_k(Q)=\frac{\bar{J}_k(Q)}{\bar{J}_k(0)}.
\end{align}

Since the boundary conditions are fixed, 
$\bar{J}_k(Q)$ and $\bar{J}_k(0)$ have the same sign as that of $\bar{\mu}_k(0)-\bar{\mu}_k(1)$ (see, e.g., \cite{ELX15, Liu18}),
and hence $\lambda_k(Q)\ge 0$. Therefore, the permanent charge $Q$ 
{\em enhances} the flux of the $k$th ion species if $\lambda_k(Q)>1$ and 
it {\em inhibits} the flux of the $k$th ion species if $\lambda_k(Q)<1$.  
Regardless the relative positions of $\lambda_i(Q)$ and $\lambda_j(Q)$ to $1$,  
we say the permanent charge $Q$ {\em enhances} the flux of $i$th ion species more than that of $j$th ion species if $\lambda_i(Q)>\lambda_j(Q)$, even if, say, $1>\lambda_i(Q)>\lambda_j(Q)$.

From \cite{JLZ15,Liu18,SL22},  we know that for $n=2$ or $n=3$, 
{\em depending on the boundary conditions}, 
either $\lambda_k(Q)\ge 1$ or $\lambda_k(Q)<1$ may occur.
In particular,
it is known (\cite{JLZ15,Liu18}) that, 
\begin{align}\label{universal}
\mbox{For }\; n=2 \;\mbox{ with }\; z_1>0>z_2, \mbox{ if }\;Q> 0,\; \mbox{ then }\; \lambda_1(Q)<\lambda_2(Q).
\end{align}
We comment that $Q$ needs not to be piecewise constant and the property holds true for any given boundary conditions.
 On the other hand,
it was shown in \cite{SL22} for a case study that, for $n=3$ with $z_1>z_2>0>z_3$,
if $Q\ge 0$, then, 
{\em dependent on boundary conditions and channel geometry}, each of the following rather surprising  situations is  possible:
(i) $\lambda_2(Q)>\lambda_3(Q)$;  
(ii) $\lambda_1(Q)>\lambda_2(Q)$ and $\lambda_1(Q)>\lambda_3(Q)$ simultaneously; 
(iii) $\lambda_1(Q)+\lambda_2(Q)>2\lambda_3(Q)$.

\subsection{Setup of our case study}\label{SetupSect}
In this paper, we continue the study on flux ratios started in \cite{SL22} to examine 
the effects of permanent charges on individual fluxes for {\em three} ion species. 

We now recall the basic setup from \cite{SL22}. Assume
\begin{itemize}
\item[(A1)] 
Consider three ion species ($n=3$) with $z_{1}>z_{2}>0>z_{3}$;

\item[(A2)] 
A piecewise constant permanent charge $Q=Q(x)$ with one nonzero region; 
that is, for a partition
$0=x_{0}<x_{1}<x_{2}<x_{3}=1$ of $[0,1]$,
\begin{align}\label{pc}
Q(x)=
\bigg\{\begin{array}{ll}
Q_{1}=Q_{3}=0, & x \in(x_{0}, x_{1}) \cup(x_{2}, x_{3}), \\ 
Q_{2}, & x \in(x_{1}, x_{2}),
\end{array}
\end{align}
where $|Q_{2}|$ is a constant small relative to $l_{k}$'s and $r_{k}$'s;

\item[(A3)] 
The electrochemical potential is ideal,
that is $\bar{\mu}_{k}=\bar{\mu}_{k}^{id}$ given by (\ref{muid});

\item[(A4)] 
Assume that $\bar{\varepsilon}_{r}(x)=1$ 
and $D_{k}(x)=D_{k}$ for some positive constants $D_{k}$. 
\end{itemize}

In the following, we will assume $\varepsilon>0$ small and treat system (\ref{PNPE1}) as a singularly perturbed system.
Then, we apply the GSP framework from \cite{EL07,Liu09} to the BVP (\ref{PNPE1})  and (\ref{PNPBV1}).
The following quantities $\alpha_j$'s for the channel geometry are crucial
\begin{align*}
\alpha_j=\frac{H(x_{j})}{H(1)}\;\mbox{ where }\; x_j's \;\mbox{ are the jump points of}\; Q(x)\;\mbox{ and }\; H(x)=\int_0^x\frac{{\rm d}s}{h(s)}.
\end{align*}


 \section{Reviews of relevant results}\label{reviews}
  \setcounter{equation}{0}
  
  We will make two reviews that are directly relevant to the present work. The first is a brief review of the geometric singular perturbation  (GSP) framework for PNP models developed in \cite{EL07, Liu09, LX15}, etc.. 
  In particular, we will recall the meromorphic function whose roots are the starting point and several critical quantities for our study of the BVP. In the second part, we review relevant results from  \cite{SL22} on flux ratios with the setup in this paper for a case study. It is the new phenomena founded in \cite{SL22} that motivate the further case study in this work. Some of the general results in \cite{SL22} are also needed for the study in this paper. 
  
  \subsection{Relevant results from GSP for PNP} 
  We will give a brief and quick account of the GSP framework and relevant results in \cite{EL07,Liu09, LX15} (with slightly different notations) 
and refer the readers to these papers and references therein for details. 
 
 Denote the derivative with respect to $x$ by overdot 
and introduce 
$u=\varepsilon \dot{\phi}$, $w_{}=x$ and 
$J_{k}={\bar{J}_{k}}/{D_{k}}$. 
System (\ref{PNPE1}) becomes, for $k=1,2,\ldots,n$,
\begin{align}
\begin{aligned}\label{ES}
\varepsilon \dot{\phi}
&=u,~~
\varepsilon \dot{u}
=-\sum_{s=1}^{n}z_{s}c_{s}
-Q(w_{})-\varepsilon \frac{h_{w_{}}(w_{})}{h(w_{})} u, \\
\varepsilon \dot{c}_{k}
&=-z_{k} c_{k} u-\frac{\varepsilon}{h(w_{})} J_{k},~~
\dot{J}_{k}=0,~~ \dot{w_{}}=1.
\end{aligned}
\end{align}
The boundary condition (\ref{PNPBV1}) becomes, 
for $k=1,2, \ldots, n$,
\begin{align*}
\phi(0)
=V,~c_{k}(0)=l_{k},~w_{}(0)=0 ;~~ 
 \phi(1)=0,~c_{k}(1)=r_{k},~w_{}(1)=1.
\end{align*}

The boundary value problem (BVP) (\ref{PNPE1}) and (\ref{PNPBV1}) can be converted to a {\em connecting orbit problem:
finding an orbit of (\ref{ES}) from $B_{L}$ to $B_{R}$} where
\begin{align}
\begin{aligned}\label{BLR}
B_{L}&=\{(\phi, u, C, J, w_{}):~ \phi=V,~C=L,~w_{}=0\}, \\ 
B_{R}&=\{(\phi, u, C, J, w_{}):~ \phi=0,~C=R,~w_{}=1\},
\end{aligned}
\end{align}
with 
$C=(c_{1}, c_{2}, \ldots, c_{n})^{T}$,
$J=(J_{1}, J_{2}, \ldots, J_{n})^{T}$,
$L=(l_{1}, l_{2}, \ldots, l_{n})^{T}$ and 
$R=(r_{1}, r_{2}, \ldots, r_{n})^{T}$.

 Due to the jumps of $Q(x)$ in (\ref{pc}) at $x_1$ and $x_2$, it is convenient to  preassign (unknown) values of $\phi$ and $c_{k}$'s 
at each jump point $x_{j}$  as
\begin{align}\label{phijcj}
\phi(x_{j})=\phi^{[j]},~~ c_{k}(x_{j})=c_{k}^{[j]},~~j=1,2.
\end{align}
We comment that, once these quantities are determined, so is a singular orbit ($\varepsilon=0$) of the BVP (\ref{PNPE1}) and (\ref{PNPBV1}) (see \cite{EL07,Liu09}).

The GSP developed in \cite{EL07,Liu09,LX15}, etc. for PNP allows one to construct three singular orbits, one over each subinterval $[x_j,x_{j+1}]$ for $j=0,1,2$ in terms of the preassigned values in (\ref{phijcj}). Each of the three orbits contains two ``boundary layers" at the end points $x_{j}$ and $x_{j+1}$, and one regular layers over the interval $(x_j,x_{j+1})$. One then matches these three orbits at $x_1$ and $x_2$ to obtain a governing system for the preassigned values in (\ref{phijcj}).
In particular, one has 
\begin{prop}[Theorem 3.1, \cite{LX15}] \label{evalue2f} Let  \begin{align}\label{Df}
D=\Gamma-I^{-1} J b^{T}
\end{align}
 where $\Gamma=\mbox{diag}\{z_1,z_2,\cdots, z_n\}$ is the diagonal matrix with $z_k$'s on the diagonal, $I=\sum_{s=1}^nz_sJ_s$, and $b=(z_{1}^{2}, z_{2}^{2}, \ldots, z_{n}^{2})^{T}$. Let
 $\sigma_1,\sigma_2,\ldots,\sigma_p$  be the distinct eigenvalues of $D$    with algebraic multiplicities  $s_1,\ldots,s_p$ so that $s_1+s_2+\ldots+s_p=n$. (Note that $0$ is always an eigenvalue and we choose $\sigma_p=0$.) 
 Then, for $j=1,2,3$, the flux over the $j$th interval is $J_{k}^{[j]}=I^{[j]}f_{k}^{[j]}$ where
\begin{align}\label{4f}\begin{split}
I^{[j]}=& \frac{V_j}{H(x_j)-H(x_{j-1})} \int_0^1 b^Te^{V_jDz}C^{[j-1,+]}\, {\rm d}z,\\
 f_{k}^{[j]}=&\frac{1}{z_k^{2}}\frac{\prod_{i=1}^p(z_k-\sigma_i)^{s_i}}{ \prod_{1\le i\le n, i\neq k}(z_k-\sigma_i)}\quad \mbox{ for }\; k=1,2,\ldots, n.
 \end{split}
 \end{align}
 \end{prop}

  For $n=3$,  the governing system for the preassigned unknowns
 $(\phi^{[j]},c_k^{[j]})$ for $j=1,2$ and $k=1,2,3$ in (\ref{phijcj}) is
\begin{align}\label{MC3}\begin{split}
  \;\mbox{ Matching at }\;  x_1: &\;  \sum_{s=1}^{3} c_{s}^{[1,-]}=\sum_{s=1}^{3} c_{s}^{[1,+]}+Q_{2}(\phi^{[1]}-\phi^{[1,+]}),\; J_{k}^{[1]}=J_{k}^{[2]},\\
  \;\mbox{ Matching at }\; x_2: &\; \sum_{s=1}^{3} c_{s}^{[2,+]}=\sum_{s=1}^{3} c_{s}^{[2,-]}+Q_{2}(\phi^{[2]}-\phi^{[2,-]}),\;J_{k}^{[2]}=J_{k}^{[3]},
\end{split}
\end{align}
where $J_{k}^{[j]}$'s are provided in Proposition \ref{evalue2f} and where,
  in terms of $(\phi^{[j]},c_k^{[j]})$ variables,  $\phi^{[j,+]}$ and $\phi^{[j,-]}$ are 
determined  by
\begin{align*}
\sum_{s=1}^{3} z_{s} c_{s}^{[j]} e^{z_{s}(\phi^{[j]}-\phi^{[j,+]})}+Q_{j+1}=0,~~
\sum_{s=1}^{3} z_{s} c_{s}^{[j]} e^{z_{s}(\phi^{[j]}-\phi^{[j,-]})}+Q_{j}=0,
\end{align*}
and
$c_{k}^{[j,+]}$ and $c_{k}^{[j,-]}$   are, in turn, given by
\begin{align*}
c_{k}^{[j,+]}=c_{k}^{[j]} e^{z_{k}(\phi^{[j]}-\phi^{[j,+]})},~~
c_{k}^{[j,-]}=c_{k}^{[j]} e^{z_{k}(\phi^{[j]}-\phi^{[j,-]})}.
\end{align*}

 \subsection{Relevant results for small $Q$ and $n=3$ from \cite{SL22} }  
 In \cite{SL22}, for three ion species with valences $z_1>z_2>0>z_3$ and for small  $|Q_2|$, a systematic investigation on comparative effects of permanent charge $Q$ was initiated and  a rich set of phenomena was revealed, much more beyond the case with two ion species done in \cite{JLZ15, Liu18}.
  Even though the study in \cite{SL22}  was focused on  special cases, many findings were quite counterintuitive. We will review the framework set in \cite{SL22} and some relevant results in this part.

  \subsubsection{Comparative effects of permanent charge $Q$ with small $|Q_2|$}\label{KVtheta}
   In \cite{SL22} and this paper, the main focus is to study the effects of 
small permanent charges on individual fluxes and we are mainly interested in properties based on $J_k$ up to $O(Q_2)$.

  For small $Q_2$, we expand the quantities $(\phi^{[j]}, c_{k}^{[j]})$ in (\ref{phijcj}) as follows
\begin{align}\label{expan}
\begin{split}
\phi^{[j]}
&=\phi^{[j]}_{0}+\phi^{[j]}_{1}Q_{2}+O(Q_{2}^{2}),\\
c_{k}^{[j]}
&=c_{k0}^{[j]}+c_{k1}^{[j]}Q_{2}+O(Q_{2}^{2}),\\
J_{k}
&=J_{k0}+J_{k1}Q_{2}+O(Q_{2}^{2}).
\end{split}
\end{align}

With expansions  in (\ref{expan}), the flux ratio $\lambda_k(Q)$ of the $k$-th ion species in (\ref{fluxratio}) and the flux ratio difference $\lambda_i(Q)-\lambda_j(Q)$ between the $i$-th and the $j$-th ion species are
\begin{align}\label{DiffFR}\begin{split}
&\lambda_k(Q)=\frac{J_{k}(Q)}{J_k(0)}=1+\tau_k Q_2+o(Q_2),\\
&\lambda_i(Q)-\lambda_j(Q)=\tau_{ij}Q_2+o(Q_2),
\end{split}
\end{align}
where, from (\ref{expan}),
\begin{align}\label{1stLambdas}
\tau_k=\frac{J_{k1}}{J_{k0}} \;\mbox{ and }\; \tau_{ij}=\tau_i-\tau_j.
\end{align}
Therefore, {\em comparative effects of small $Q$} on fluxes $J_i$ and $J_j$  are reduced to the study of signs of  
$\tau_{ij}$, 
which are determined by boundary conditions $(V,L,R)$ and channel geometry characters $(\alpha_1,\alpha_2)$.
\medskip

 We will use the following notations for simplicity.
\begin{align}
\begin{aligned}\label{notation}
&S_L=\sum_{s=1}^{3}l_{s}, ~~ S_R=\sum_{s=1}^{3}r_{s}, ~~
\Lambda_L=\sum_{s=1}^{3}z_{s}^{2}l_{s}, ~~\Lambda_R=\sum_{s=1}^{3}z_{s}^{2}r_{s},\\
&m_l=-z_{1}z_{2}z_{3}\frac{S_L}{\Lambda_L},~~
m_r=-z_{1}z_{2}z_{3}\frac{S_R}{\Lambda_R},~~
\rho=\frac{\Lambda_R}{\Lambda_L}.
\end{aligned}
\end{align}

Let $D_{0}$ be the zeroth order of $D$ in $Q_{2}$. For $V\neq0$, under the boundary electroneutrality conditions, let $g: \mathbb{C} \rightarrow \mathbb{C}$ be the meromorphic function defined as
\begin{align}\label{gj}
g(\sigma):=
\sum_{s=1}^{3}\frac{~z_{s}^{2}r_{s}}{z_{s}-\sigma}
-e^{V\sigma}
 \sum_{s=1}^{3}\frac{~z_{s}^{2}l_{s}}{z_{s}-\sigma},
\end{align}
and denote
\begin{align}\label{gammafun}
\gamma(\sigma)
=g(\sigma)\frac{1}{\sigma}\prod_{s=1}^{3}(\sigma-z_{s})
=e^{\sigma V}L(\sigma)-R(\sigma),
\end{align}
where 
\begin{align*}
L(\sigma)=z_{1}z_{2}z_{3}S_L+\sigma\Lambda_L=\Lambda_L(\sigma-m_l),~~
R(\sigma)=z_{1}z_{2}z_{3}S_R+\sigma\Lambda_R=\Lambda_R(\sigma-m_r).
\end{align*}
 
As a direct consequence of a general result in  \cite{LX15}, one has that $D_0$ has a zero eigenvalue and its nonzero eigenvalues $\sigma_{10}$ and $\sigma_{20}$ are the unique solutions of $\gamma(\sigma)=0$
in  the stripe
\[S=\Big\{z\in \C: \mbox{Im}(z) \in (-\pi/|V|, \pi/|V|)\Big\}.\]

\bigskip

We now recall several results from \cite{SL22} that will be used in this work.
\begin{prop}[Proposition 6.3, \cite{SL22}]\label{Lk0}
Suppose $V\neq 0$.

 (i) If $\sigma_{10}\sigma_{20}\neq0$, 
then $\phi^{[j]}_{0}$ is uniquely   determined by
\begin{align}\label{phia}
e_{0}^{T}e^{(V-\phi^{[j]}_{0})D_{0}}L-S_L=\alpha_jg'(0).
\end{align}

(ii) If $\sigma_{10}\neq0$ and $\sigma_{20}=0$,
then $\phi^{[j]}_{0}$ is uniquely   determined by
\begin{align*}
 e_{0}^{T}\Gamma^{-1}e^{(V-\phi_0^{[j]})D_{0}}-\sum_{s=1}^{3}\frac{l_s}{z_s}-(V-\phi_0^{[j]})S_L
=\alpha_j g''(0).
\end{align*}
\end{prop}

\begin{prop}[Proposition 3.3, \cite{SL22}]\label{Tij}   One has
   \begin{align*}
 \tau_{ij}=&T_{ij}(\phi_0^{[2]})-T_{ij}(\phi_0^{[1]})
 \end{align*} 
 where, for a permutation $\{i,j,k\}$ of $\{1,2,3\}$ (the convention to be used in the rest),
\begin{itemize}
\item[(i)] if $\sigma_{10}\neq \sigma_{20}$, then 
\begin{align}\label{sigmaijphi}
T_{ij}(\phi)
=\frac{(z_{i}-z_{j})(\sigma_{10}-z_{k})}
      {\sigma_{10}\gamma'(\sigma_{10})}(e^{\sigma_{10}\phi}-1)
+\frac{(z_{i}-z_{j})(\sigma_{20}-z_{k})}
      {\sigma_{20}\gamma'(\sigma_{20})}(e^{\sigma_{20}\phi}-1);
\end{align}
(If $\sigma_{j0}=0$, then the above formula is defined by applying L'Hopital rule.)
\item[(ii)] if $\sigma_{10}=\sigma_{20}=\sigma_0$, then
\begin{align*}
T_{ij}(\phi)
=\frac{2(z_{i}-z_{j})(\sigma_{0}-z_{k})}
      {\sigma_{0}\gamma''(\sigma_{0})}
\Big(\phi e^{\sigma_{0}\phi}-(e^{\sigma_{0}\phi}-1)
       \big(\frac{\gamma'''(\sigma_{0})}
                  {3\gamma''(\sigma_{0})}
     +\frac{1}{\sigma_{0}}+\frac{1}{z_{k}-\sigma_{0}}\big)\Big).
\end{align*}
(If $\sigma_{0}=0$, then the above formula is defined by applying L'Hopital rule.)
\end{itemize}
\end{prop}

We recall that, if $V=0$, then $\phi^{[1]}_0=\phi^{[2]}_0=0$ (Proposition 2.1 in \cite{ELX15}), 
and hence, $\tau_{ij}=0$.  Also if $V>0$ (resp. $V<0$), 
then $\phi'(x)<0$ and $\phi_0^{[1]}>\phi_0^{[2]}$ (resp. $\phi'(x)>0$ and $\phi_0^{[1]}<\phi_0^{[2]}$).
In the sequel, we will assume $V\neq 0$ Therefore,  for $\sigma_{10},\sigma_{20}\in \R$,    the domain of $T_{ij}(\phi)$ is $\R\setminus\{0\}$;
for $\sigma_{10}=\overline{\sigma_{20}}=x+iy$ with $y\neq0$, 
the domain of $T_{ij}(\phi)$ is $(-\pi/|y|,0)\cup(0,\pi/|y|)$.

\begin{rem}\label{zeroV}
(i) It follows from Proposition 3.5 in \cite{SL22} that for $V\neq 0$,  the situation $\tau_{12}=\tau_{23}=\tau_{13}=0$ cannot occur.

(ii) For a fixed boundary condition, $\phi_0^{[j]}$ depends on $\alpha_j$,
so we will use $\tau_{ij}(\alpha_1,\alpha_2)$ whenever needed to 
emphasize the dependance of $\tau_{ij}$ on $(\alpha_1,\alpha_2)\in \Delta$,
where 
\[\Delta:=\{(\alpha_1,\alpha_2)\mid0 \leq \alpha_1 \leq \alpha_2 \leq 1\}.\]
\end{rem}

\section{New case study  on $\tau_{ij}$ for $\sigma_{10}=\overline{\sigma_{20}}=x+iy$}\label{SS2}
\setcounter{equation}{0}
In \cite{SL22}, the authors presented preliminary results for
$\sigma_{10}\neq \sigma_{20}\in \R$ and $\sigma_{10}=\sigma_{20}\in \R$
and applied to study the equi-chemical-potential-difference case.
Here we give some preliminary results for $\sigma_{10}=\overline{\sigma_{20}}\in \C$.

\subsection{Critical points of $T_{ij}(\phi)$, the quantity $\theta_{ij}$ and the function $P_{ij}(\alpha)$.}
In this case, the function $T_{ij}(\phi)$ in (\ref{sigmaijphi}) is still real-value 
and we will make a preparation before a statement of the result. 
By Proposition \ref{evalue2f}, we know $y\in(-\pi/|V|,\pi/|V|)$.  
For $\sigma_{10}=\overline{\sigma_{20}}=x+iy$, one has, from (\ref{gammafun}),
\begin{align*}
e^{(x+iy)V}
=\frac{\Lambda_R}{\Lambda_L}\frac{x-m_r+iy}{x-m_l+iy},
\end{align*}
which leads to
\begin{align}\label{scx}
e^{xV}\cos(yV)
=\rho
\frac{y^{2}+(x-m_l)(x-m_r)}
     {y^{2}+(x-m_l)^{2}}\;\mbox{ and }\;
e^{xV}\sin(yV)
=\rho
\frac{y(m_r-m_l)}
     {y^{2}+(x-m_l)^{2}}.
\end{align}

Thus,
\[e^{2xV}=\rho^2\frac{y^2+(x-m_r)^2}{y^2+(x-m_l)^2},\;\mbox{ and hence,}\;
y^2=\frac{\rho^2(x-m_r)^2-e^{2xV}(x-m_l)^2}{e^{2xV}-\rho^2}.\]

\begin{lem}\label{veta} If  $\sigma_{10}=\overline{\sigma_{20}}=x+iy$ with $y\neq0$,
then $V(m_r-m_l)>0$.
\end{lem}
\begin{proof} It is a consequence of the 2nd equation in (\ref{scx}) and $yV\sin(yV)>0$.
\end{proof}

It follows from (\ref{gammafun}) that $\gamma'(x+iy)=\zeta+i\kappa$ where
\begin{align}
\begin{aligned}\label{c2d2}
\zeta
&=\Lambda_R
\frac{(m_l-x)(m_r-m_l)+V(x-m_r)(y^2+(x-m_l)^{2} )}{y^2+(x-m_l)^{2}},\\
\kappa
&=\Lambda_R
\frac{y(m_r-m_l)+yV(y^{2}+(x-m_l)^{2})}{y^2+(x-m_l)^2}.
\end{aligned}
\end{align}


\begin{lem}
Assume  $\sigma_{10}=\overline{\sigma_{20}}=x+iy$ with $y\neq0$.
 Over $(-\pi/|y|,\pi/|y|)$, one has
\begin{itemize}
\item[(i)] if $y\kappa+(x-z_{k})\zeta=0$, then $\phi=0$ is  the unique critical point of $T_{ij}(\phi)$;

\item[(ii)]
if $y\kappa+(x-z_{k})\zeta\neq 0$,  then $T_{ij}(\phi)$ has 
two critical points, one in $(-\frac{\pi}{|y|},0)$ and the other in $(0,\frac{\pi}{|y|})$.
\end{itemize}
\end{lem}
\begin{proof} 
It follows from (\ref{sigmaijphi}) that
\begin{align*}
T'_{ij}(\phi)
=2(z_{i}-z_{j})e^{x\phi}
\frac{\cos(y\phi)(y\kappa+(x-z_{k})\zeta)
     -\sin(y\phi)(y\zeta-(x-z_{k})\kappa)}
     {\zeta^{2}+\kappa^{2}}.
\end{align*}

(i) If $y\kappa+(x-z_{k})\zeta=0$,  then  $\zeta\neq0$ since, 
otherwise, $\kappa=0$ too that contradicts to $\zeta+i\kappa\neq 0$. One then gets  
\begin{align}\label{DTij2}
T'_{ij}(\phi)=-2(z_{i}-z_{j})e^{x\phi}\sin(y\phi)\frac{y}{\zeta}.
\end{align}
Thus $T'_{ij}(\phi)=0$ has a unique root $\phi=0$. 
 
 (ii) If $y\kappa+(x-z_{k})\zeta\neq0$, then $T'_{ij}(\phi)=0$ implies $\sin(y\phi)\neq 0$, and hence, 
\[\cot(y\phi)=\frac{y\zeta-(x-z_{k})\kappa}{y\kappa+(x-z_{k})\zeta},\]
which has two roots  in $(-\frac{\pi}{|y|},\frac{\pi}{|y|})$, one in $(-\frac{\pi}{|y|},0)$ and the other in $(0,\frac{\pi}{|y|})$.
\end{proof}

For $y\kappa+(x-z_{k})\zeta\neq0$,  we denote the two critical points of $T_{ij}(\phi)$ in $(-\pi/{|y|},\pi/{|y|})$ by $V_{ij}$  and $\tilde{V}_{ij}$,  and set
\begin{align*}
\theta_{ij}=\frac{1}{g'(0)}e_{0}^{T}\Big(e^{(V-V_{ij})D_{0}}L-L\Big)\;\mbox{ and }\;
\tilde{\theta}_{ij}=\frac{1}{g'(0)}e_{0}^{T}\Big(e^{(V-\tilde{V}_{ij})D_{0}}L-L\Big).
\end{align*}
Without loss of generality, we assume $V\tilde{V}_{ij}<0<VV_{ij}$.

\begin{figure}[hpbt]
\centering
\includegraphics[width=0.45\linewidth]{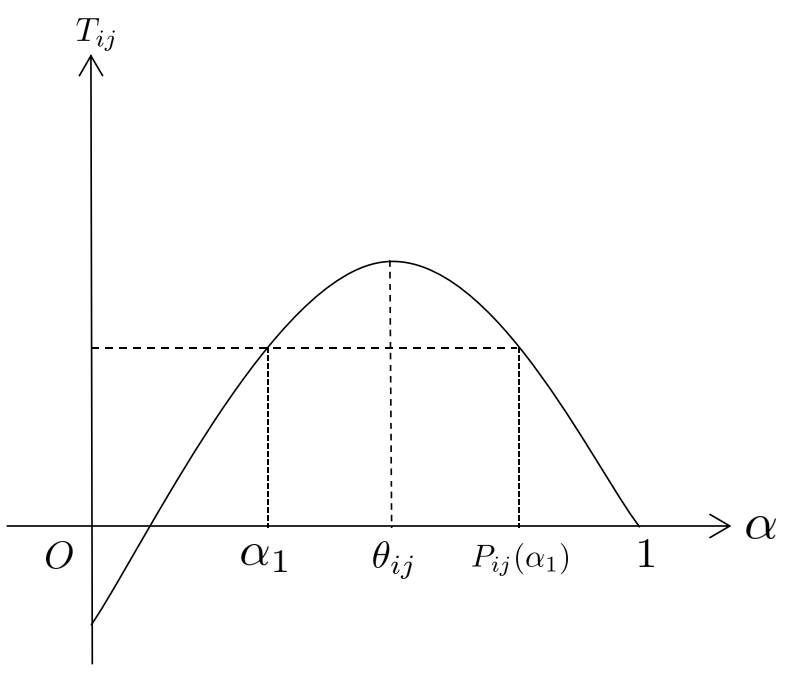} 
\caption{\em Function $P_{ij}(\alpha)$ from $T_{ij}(\alpha)$}
\label{FigPij}
\vspace{-0.5em}
\end{figure}

 \begin{prop}\label{thetainV}
If $V>0$ (resp. $V<0$), then $\theta_{ij}$ and $\tilde{\theta}_{ij}$ 
are decreasing (resp. increasing) in  $V_{ij}$ and $\tilde{V}_{ij}$, respectively.
Furthermore,  $\theta_{ij}<1<\tilde{\theta}_{ij}$ holds true,
and one has $\theta_{ij}\in [0,1]$ if and only if $V_{ij}$ lies between $0$ and $V$. 
\end{prop}

When $\theta_{ij}\in [0,1]$,  we can define a function $P_{ij}:(-\infty,\theta_{ij}]\to [\theta_{ij},\infty)$ as follows: 
$P_{ij}(\alpha)=\beta$ if  $T_{ij}(\beta)=T_{ij}(\alpha)$ (see Figure \ref{FigPij}). This function plays important role in characterizing transitions from one behavior of permanent charge effect to another.

\subsection{Signs of $\tau_{ij}$}\label{genComplex}
In this section, we will discuss signs of $\tau_{ij}$ for
$\sigma_{10}=\overline{\sigma_{20}}=x+iy$ with $y\neq0$,
especially the case with $x=0$.

\subsubsection{General results}
The following result gives the conditions of determining signs of $\tau_{ij}$,
which are determined by $(V,L,R)$ and $(\alpha_1,\alpha_2)$. 
As in Remark \ref{zeroV},  we will use $\tau_{ij}(\alpha_1,\alpha_2)$ 
whenever needed to emphasize the dependance of $\tau_{ij}$ on 
$(\alpha_1,\alpha_2) \in \Delta$.


\begin{prop}\label{KVij2}
Assume $\sigma_{10}=\overline{\sigma_{20}}=x+iy$ with $y\neq0$.
For $V\neq0$, one has
\begin{itemize}
\item[(i)]
if $y\kappa+(x-z_{k})\zeta=0$,
then $(z_{i}-z_{j})\tau_{ij}(\alpha_1,\alpha_2)\zeta>0$ for any $(\alpha_1,\alpha_2) \in \Delta$;

\item[(ii)]
if $y\kappa+(x-z_{k})\zeta\neq0$ and $\theta_{ij}\in [0,1]$, 
then exactly one of the followings occurs
\begin{itemize}
\item[(a)]
when $(z_{i}-z_{j})V(y\kappa+(x-z_{k})\zeta)<0$, 
\[\tau_{ij}(\alpha_1,\alpha_2)
=\left\{\begin{array}{ll}
<0,~~ \alpha_1 <\theta_{ij}\mbox{ and }~\alpha_2<P_{ij}(\alpha_1),\\
>0,~~ \alpha_1>\theta_{ij}\mbox{ or }~ \alpha_2> P_{ij}(\alpha_1);\end{array}\right.\]

\item[(b)]
when $(z_{i}-z_{j})V(y\kappa+(x-z_{k})\zeta)>0$, 
\[\tau_{ij}(\alpha_1,\alpha_2)
=\left\{\begin{array}{ll}
>0,~~ \alpha_1 <\theta_{ij}\mbox{ and }~\alpha_2<P_{ij}(\alpha_1),\\
<0,~~ \alpha_1>\theta_{ij}\mbox{ or }~ \alpha_2> P_{ij}(\alpha_1).\end{array}\right.\]
\end{itemize}
\end{itemize}
\end{prop}
\begin{proof}
For $V\neq0$, it can be shown that 
$\tau_{ij}(\alpha,\alpha)=0$ for any $\alpha\in[0,1]$
and
\begin{align}\label{pd1}
\partial_{\alpha_1}\tau_{ij}(\alpha_1,\alpha_2) 
=-\frac{{\rm d}\phi^{[1]}_{0}}{{\rm d}\alpha_1}
T'_{ij}(\phi^{[1]}_{0}),~~
\partial_{\alpha_2}\tau_{ij}(\alpha_1,\alpha_2)
=\frac{{\rm d}\phi^{[2]}_{0}}{{\rm d}\alpha_2}
T'_{ij}(\phi^{[2]}_{0}).
\end{align}
Note that 
for $V>0$, $\frac{{\rm d}\phi^{[1]}_{0}}{{\rm d}\alpha_1}, \frac{{\rm d}\phi^{[2]}_{0}}{{\rm d}\alpha_2}<0$;
for $V<0$, $\frac{{\rm d}\phi^{[1]}_{0}}{{\rm d}\alpha_1}, \frac{{\rm d}\phi^{[2]}_{0}}{{\rm d}\alpha_2}>0$.

(i) Suppose $y\kappa+(x-z_{k})\zeta=0$.
Since $y\phi\in(-\pi,\pi)$, one has $y\phi\sin(y\phi)>0$ if $\phi\neq 0$. It then follows from (\ref{DTij2}) that 
\begin{align}\label{signTd}
(z_{i}-z_{j})\zeta\phi T'_{ij}(\phi)< 0\;\mbox{ for }\; \phi\neq 0.
\end{align}
 Recall that $\tau_{ij}=T_{ij}(\phi_0^{[2]})-T_{ij}(\phi_0^{[1]})$. Thus, for some $\phi^*$ between $\phi_0^{[1]}$ and $\phi_0^{[2]}$,
\begin{align}\label{pstar}
(z_{i}-z_{j})\zeta\phi^*\tau_{ij}
=(z_{i}-z_{j})\zeta\phi^* T_{ij}'(\phi^*)(\phi_0^{[2]}-\phi_0^{[1]}).
\end{align}
For $V>0$, one has $\phi_0^{[1]}>\phi_0^{[2]}>0$, and hence, $\phi^*>0$. 
It then follows from (\ref{signTd}) and (\ref{pstar}) that
$(z_{i}-z_{j})\zeta\tau_{ij}>0$. Similarly, for $V<0$, one also has $(z_{i}-z_{j})\zeta\tau_{ij}>0$.

%
\medskip

(ii) Suppose that $y\kappa+(x-z_{k})\zeta\neq0$ and $\theta_{ij}\in [0,1]$.
For ease of the proof, we set
\[V_{ij}^{1}=V_{ij},~~V_{ij}^{2}=\tilde{V}_{ij},~~
\theta_{ij}^{1}=\theta_{ij},~~\theta_{ij}^{2}=\tilde{\theta}_{ij}.\]
Note that $0\le \theta_{ij}^{1}\le 1<\theta_{ij}^{2}$.
Since $T'_{ij}(V_{ij}^{s})=0$, $s=1,2$, a direct calculation gives
\begin{align*}
\begin{aligned}
\partial_{\alpha_1\alpha_1}\tau_{ij}(\theta_{ij}^{s},\theta_{ij}^{s})
=-\left(\frac{{\rm d}\phi^{[1]}_{0}}{{\rm d}\alpha_1}\right)^{2}
T''_{ij}(V_{ij}^{s}),~~
\partial_{\alpha_2\alpha_2}\tau_{ij}(\theta_{ij}^{s},\theta_{ij}^{s})
=\left(\frac{{\rm d}\phi^{[2]}_{0}}{{\rm d}\alpha_2}\right)^{2}
T''_{ij}(V_{ij}^{s}),
\end{aligned}
\end{align*}
where
\begin{align*}
T''_{ij}(V_{ij}^{s})
=-\frac{2(z_{i}-z_{j})y\sin(yV_{ij}^{s})e^{xV_{ij}^{s}}}{\zeta^{2}+\kappa^{2}}
  \frac{(y\kappa+(x-z_{k})\zeta)^{2}+(y\zeta-(x-z_{k})\kappa)^{2}}{y\kappa+(x-z_{k})\zeta}.
\end{align*}
By \[yV_{ij}^{1}\in(-\pi,\pi),~~yV_{ij}^{2}\in(-\pi,\pi),~~V_{ij}^{1}V_{ij}^{2}<0\]
one has $\sin(yV_{ij}^{1})\sin(yV_{ij}^{2})<0$, which leads to
$T''_{ij}(V_{ij}^{1})T''_{ij}(V_{ij}^{2})<0$.
It follows from
\begin{align*}
\partial_{\alpha_1\alpha_1}\tau_{ij}(\theta_{ij}^{s},\theta_{ij}^{s})
\partial_{\alpha_2\alpha_2}\tau_{ij}(\theta_{ij}^{s},\theta_{ij}^{s})<0,~~
\partial_{\alpha_1\alpha_2}\tau_{ij}(\alpha_1,\alpha_2)
=\partial_{\alpha_2\alpha_1}\tau_{ij}(\alpha_1,\alpha_2)=0
\end{align*}
that $(\theta_{ij}^{1},\theta_{ij}^{1})\in\Delta$ and $(\theta_{ij}^{2},\theta_{ij}^{2})\notin\Delta$
are two saddle points of  $\tau_{ij}(\alpha_1,\alpha_2)$.

When $(z_{i}-z_{j})V(y\kappa+(x-z_{k})\zeta)<0$, 
one has  $T''_{ij}(V_{ij}^{1})>0>T''_{ij}(V_{ij}^{2})$.
For any $\tilde{\alpha}_{1}<\tilde{\alpha}_{2}\leq\theta_{ij}^{1}
\leq\tilde{\beta}_{1}<\tilde{\beta}_{2}\leq\theta_{ij}^{2}
\leq\tilde{\eta}_{1}<\tilde{\eta}_{2}$,
it follows from
\begin{align*}\partial_{\alpha_1\alpha_1}\tau_{ij}(\theta_{ij}^{1},\theta_{ij}^{1})
<0<\partial_{\alpha_2\alpha_2}\tau_{ij}(\theta_{ij}^{1},\theta_{ij}^{1}),~~
\partial_{\alpha_1\alpha_1}\tau_{ij}(\theta_{ij}^{2},\theta_{ij}^{2})
>0>\partial_{\alpha_2\alpha_2}\tau_{ij}(\theta_{ij}^{2},\theta_{ij}^{2})
\end{align*}
that
\[\partial_{\alpha_{1}}\tau_{ij}(\tilde{\alpha}_{1},\tilde{\alpha}_{2})>0,~~
\partial_{\alpha_{1}}\tau_{ij}(\tilde{\beta}_{1},\tilde{\beta}_{2})<0<
\partial_{\alpha_{2}}\tau_{ij}(\tilde{\beta}_{1},\tilde{\beta}_{2}),~~
\partial_{\alpha_{2}}\tau_{ij}(\tilde{\eta}_{1},\tilde{\eta}_{2})<0.\]
Thus, \[\tau_{ij}(\tilde{\alpha}_{1},\tilde{\alpha}_{2})
<0<\tau_{ij}(\tilde{\beta}_{1},\tilde{\beta}_{2}),~~
\tau_{ij}(\tilde{\eta}_{1},\tilde{\eta}_{2})<0.\]
Similarly, when $(z_{i}-z_{j})V(y\kappa+(x-z_{k})\zeta)>0$, one has
\[\tau_{ij}(\tilde{\alpha}_{1},\tilde{\alpha}_{2})
>0>\tau_{ij}(\tilde{\beta}_{1},\tilde{\beta}_{2}),~~
\tau_{ij}(\tilde{\eta}_{1},\tilde{\eta}_{2})>0.\]
Moreover, it follows from the implicit function theorem that
$\tau_{ij}(\alpha_1,\alpha_2)=0$ has a solution 
$\alpha_2=P_{ij}(\alpha_1)\ge \theta_{ij}^{1}$ for $\alpha_1\in [0,\theta_{ij}^{1}]$ 
satisfying $\theta_{ij}^{1}=P_{ij}(\theta_{ij}^{1})$ (see Fig. \ref{FigPij}).
The proof is completed.
\end{proof}

\subsubsection{Special case with $\sigma_{10}=\overline{\sigma_{20}}=iy$ with $y>0$}\label{SPcs}
In this part, we will conduct a detailed study for the case with
$\sigma_{10}=\overline{\sigma_{20}}=iy$ with $y>0$, more precisely, $0<y<\pi/|V|$.
It follows from (\ref{scx}) that
\begin{align}\label{cs}
\cos(yV)=\rho\frac{y^{2}+m_lm_r}{y^{2}+m_l^{2}}\;\mbox{ and }\;
\sin(yV)=\rho\frac{y(m_r-m_l)}{y^{2}+m_l^{2}}.
\end{align}
By the symmetry in Remark \ref{swapsym} and Lemma \ref{veta},  
we only consider $m_l<m_r$, or equivalently, $V>0$.
It is direct to obtain from $yV\in(0,\pi)$ and (\ref{cs})  that $yV\in(0,\frac{\pi}{2})$.

%

\begin{lem}\label{VyL}
Assume $m_l<m_r$. 
$\sigma_{10}=\overline{\sigma_{20}}=iy$ with $y>0$ is the solution of $\gamma(\sigma)=0$ if and only if
\begin{align}\label{Vycot}
V=\frac{1}{y}
\cot^{-1} \left(\frac{\rho^{2}m_r+m_l}
{\sqrt{\left(\rho^{2} m_r^{2}-m_l^{2}\right)\left(1-\rho^{2}\right)}}\right)\in(0,\frac{\pi}{2y}),
\end{align}
where
\begin{align*}
y
=\sqrt{\frac{\rho^{2} m_r^{2}-m_l^{2}}{1-\rho^{2}}}\;\mbox{ and }\;
\rho\in(\frac{m_l}{m_r},1).
\end{align*}
\end{lem}
\begin{proof}
By (\ref{cs}), one has
\[\rho^2\frac{(y^{2}+m_lm_r)^2+y^2(m_r-m_l)^2}{(y^{2}+m_l^{2})^2} =1\Longrightarrow \rho^2(y^{2}+m_r^{2})=y^{2}+m_l^{2}\Longrightarrow y^2=\frac{\rho^2m_r^2-m_l^2}{1-\rho^2}.\]
It is easy to see that the right-hand side is positive if and only if  $m_l/m_r<\rho<1$
which, in turn, is equivalent to $S_L<S_R$ and $\Lambda_R<\Lambda_L$. 
The formula (\ref{Vycot}) follows from
\begin{align*}
y=\sqrt{\frac{\rho^{2} m_r^{2}-m_l^{2}}{1-\rho^{2}}}>0\;\mbox{ and }\;
\cot(yV)
=\frac{y^{2}+m_lm_r}{y(m_r-m_l)}
=\frac{\rho^{2}m_r+m_l}
	   {\sqrt{\left(\rho^{2} m_r^{2}-m_l^{2}\right)\left(1-\rho^{2}\right)}}>0.
\end{align*}  
The proof is completed.
\end{proof}

By (\ref{c2d2}) and $\sigma_{10}=\overline{\sigma_{20}}=iy$, one has
\begin{align}\label{lk}
\zeta
=\Lambda_R\frac{m_l(m_r-m_l)-m_rV(y^{2}+m_l^{2})}{y^{2}+m_l^{2}},~~
\kappa
=\Lambda_R\frac{y(m_r-m_l)+yV(y^{2}+m_l^{2})}{y^{2}+m_l^{2}}.
\end{align}
Let
\begin{align*}
V_{0}
=\frac{m_l(m_r-m_l)}{m_r(y^{2}+m_l^{2})}.
\end{align*}

\begin{lem}\label{cd1}
Assume $m_l<m_r$ and $\sigma_{10}=\overline{\sigma_{20}}=iy$ with $y>0$.
One has 
\[V>V_{0},~~\kappa>0>\zeta,~~y\zeta+m_r\kappa>0.\]
\end{lem}
\begin{proof}
It follows from $yV\in(0,\frac{\pi}{2})$ and  $\rho\in(\frac{m_l}{m_r},1)$ that
\begin{align*}
yV_{0}-yV
=\frac{m_l\sin(yV)}{m_r\rho}-yV<\sin(yV)-yV<0,
\end{align*}
and hence $V>V_{0}$ by $y>0$.
It is direct to obtain from $m_l<m_r$ and (\ref{lk}) that $\kappa>0$, 
and $\zeta<0$ by $V>V_{0}$.
Note that 
\begin{align*}
\frac{\zeta}{\kappa}
=\frac{m_l(m_r-m_l)-m_rV(y^{2}+m_l^{2})}
     {y(m_r-m_l)+yV(y^{2}+m_l^{2})}<0.
\end{align*}
It can be shown from $V>0$ that $\frac{\zeta}{\kappa}\in(-\frac{m_r}{y},0)$,
which leads to $y\zeta+m_r\kappa>0$.
The proof is completed.
\end{proof}

For $y\kappa-z_{k}\zeta\neq0$, let
\begin{align}\label{cotV}
\theta_{ij}=\frac{1}{g'(0)}e_{0}^{T}\left(e^{(V-V_{ij})D_{0}}L-L\right),~~
V_{ij}=\frac{1}{y}\cot^{-1}\left(\frac{y\zeta+z_{k}\kappa}{y\kappa-z_{k}\zeta}\right)\in(0,\frac{\pi}{y}). 
\end{align}
One has from Proposition \ref{thetainV} that $\theta_{ij}<1$. 
In the following, we will determine conditions for $\theta_{ij}>0$. 

It follows from Proposition \ref{thetainV},
$yV\in(0,\frac{\pi}{2})$ and $yV_{ij}\in(0,\pi)$ that
$\theta_{ij}>0$ if and only if $V_{ij}<V$ if and only if $\cot(yV)<\cot(yV_{ij})$.
Since $\cot(yV)>0$, 
we first determine conditions for $\cot(yV_{ij})>0$, or equivalently, 
for $(y\zeta+z_{k}\kappa)(y\kappa-z_{k}\zeta)>0$.
Note that $z_{3}<0<z_{2}<m_l<m_r<z_{1}$.
One has from Lemma \ref{cd1} that
\begin{align}\label{ykz23}
y\kappa-z_{1}\zeta>y\kappa-z_{2}\zeta>0, ~~
y\zeta+z_{1}\kappa>0>y\zeta+z_{3}\kappa.
\end{align}

%
%

\begin{lem}\label{z3} 
Assume $m_l<m_r$ and $\sigma_{10}=\overline{\sigma_{20}}=iy$ with $y>0$.
One has

(i) $y\kappa-z_3\zeta<0$ if and only if 
\begin{align}\label{eqV}
\begin{split}
y^2+z_3m_r<0\;\mbox{ and }\;
V>\frac{(z_3m_l-y^2)(m_r-m_l)}{(z_3m_r+y^2)(y^2+m_l^2)}.
\end{split}
\end{align}

(ii) $y\zeta+z_2\kappa>0$ if and only if 
\[V<\frac{(m_l+z_2)(m_r-m_l)}{(m_r-z_2)(y^2+m_l^2)}.\]
\end{lem}
\begin{proof}    
(i) Using  (\ref{lk}), one  has  that $y\kappa-z_3\zeta<0$ if and only if 
\[(y^2-z_3m_l)(m_r-m_l)+V(y^2+z_3m_r)(y^2+m_l^2)<0.\]
If $y^2+z_3m_r\ge 0$, then the above fails since $z_3<0$, $V>0$ and $m_l<m_r$.
If $y^2+z_3m_r<0$, then the above is equivalent to  
\[V>\frac{(z_3m_l-y^2)(m_r-m_l)}{(z_3m_r+y^2)(y^2+m_l^2)}.\]

(ii) Using  (\ref{lk}), one  has  that $y\zeta+z_2\kappa>0$ if and only if 
\[y(m_l+z_2)(m_r-m_l)-Vy(m_r-z_2)(y^2+m_l^2)>0.\]
By $m_r>z_2$ and $y>0$, the above is equivalent to
\[V<\frac{(m_l+z_2)(m_r-m_l)}{(m_r-z_2)(y^2+m_l^2)}.\]
The proof is completed.
\end{proof}

%
%

We now determine conditions for $\theta_{ij}>0$, or equivalently,
for the sign of $\cot(yV)-\cot(yV_{ij})$.
It is direct to obtain
\begin{align}\label{cotD}
\cot(yV)-\cot(yV_{ij})
=\frac{M_{ij}(V-V^*_{ij})(y^{2}+m_l^{2})}
      {y(y\kappa-z_{k}\zeta)(m_r-m_l)},
\end{align}
where
\begin{align*} 
M_{ij}
=(z_{k}m_l+y^{2})(y^{2}+m_r^{2})\;\mbox{ and }\;
V^*_{ij}
=\frac{(z_{k}m_r-y^{2})(m_r-m_l)}{M_{ij}}.
\end{align*}

\begin{lem}\label{KVij2*}
Assume $m_l<m_r$ and $\sigma_{10}=\overline{\sigma_{20}}=iy$ with $y>0$. 
One has 
\begin{itemize}
\item[(i)]

$\theta_{12}> 0$  if and only if $y\kappa-z_{3}\zeta<0$ and $M_{12}(V-V^*_{12})> 0$;

\item[(ii)]

$\theta_{13}> 0$  if and only if $y\zeta+z_{2}\kappa>0$ and $V<V^*_{13}$;

\item[(iii)]
$\theta_{23}>0$ if and only if $V<V^*_{23}$.
\end{itemize}
\end{lem}
\begin{proof}
It follows from Proposition \ref{thetainV},
$yV\in(0,\frac{\pi}{2})$ and $yV_{ij}\in(0,\pi)$ that
$\theta_{ij}>0$ if and only if $V_{ij}<V$ if and only if $\cot(yV)<\cot(yV_{ij})$.
Note that $\cot(yV)>0$.

(i) Note that $y\zeta+z_{3}\kappa<0$. 
If $y\kappa-z_{3}\zeta>0$, then it follows from (\ref{cotV}) that 
$\cot(yV_{12})<0$, which leads to $\theta_{12}<0$.
If $y\kappa-z_{3}\zeta<0$, then it follows from (\ref{cotD}) that 
$\cot(yV)<\cot(yV_{12})$ if and only if $M_{12}(V-V^*_{12})> 0$.

(ii) Note that $y\kappa-z_{2}\zeta>0$.
If $y\zeta+z_{2}\kappa<0$, then it follows from (\ref{cotV}) that 
$\cot(yV_{13})<0$, which leads to $\theta_{13}<0$.
If $y\zeta+z_{2}\kappa>0$, then it follows from (\ref{cotD}) and $M_{13}>0$ that 
$\cot(yV)<\cot(yV_{13})$ if and only if $V<V^*_{13}$.

(iii) Note that $y\kappa-z_{1}\zeta>0$ and $y\zeta+z_{1}\kappa>0$.
It follows from (\ref{cotD})  and $M_{23}>0$ 
that $\cot(yV)<\cot(yV_{23})$ if and only if $V<V^*_{23}$.
The proof is completed.
\end{proof}

{
\begin{rem}
(i) Note that $y^2+z_3m_r<y^2+z_3m_l$ and
\[V^*_{12}-\frac{(z_3m_l-y^2)(m_r-m_l)}{(z_3m_r+y^2)(y^2+m_l^2)}
=\frac{y^2(y^2+z_3^2)(m_r+m_l)(m_r-m_l)^2}
{(y^2+z_3m_r)(y^2+z_3m_l)(y^{2}+m_r^{2})(y^2+m_l^2)}.\]
By (i) in Lemma \ref{z3} and (i) in Lemma \ref{KVij2*},
it can be shown that $\theta_{12}> 0$ if and only if
either
\[ y^2+z_3m_r<y^2+z_3m_l<0~~\text{and}~~ 
V^*_{12}>V>\frac{(z_3m_l-y^2)(m_r-m_l)}{(z_3m_r+y^2)(y^2+m_l^2)},\]
or 
\[ y^2+z_3m_r<0<y^2+z_3m_l~~\text{and}~~ V>V^*_{12}.\]


(ii) Note that
\[V^*_{13}-\frac{(m_l+z_2)(m_r-m_l)}{(m_r-z_2)(y^2+m_l^2)}
=\frac{(y^{2}+z_2^{2})(y^{2}+m_lm_r)(m_l^{2}-m_r^{2})}
{(m_r-z_2)(y^2+z_2m_l)(y^2+m_r^2)(y^2+m_l^2)}<0.\]
By (ii) in Lemma \ref{z3}, we can rewrite (ii) in Lemma \ref{KVij2*} as
$\theta_{13}> 0$  if and only if $V<V^*_{13}$.
\end{rem}}

The next result gives the order of $\theta_{12}$, $\theta_{13}$ and $\theta_{23}$.

\begin{lem}\label{V13V}
Assume $m_l<m_r$ and $\sigma_{10}=\overline{\sigma_{20}}=iy$ with $y>0$. 
If $y\kappa-z_{3}\zeta>0$, then $\theta_{12}<\theta_{13}<\theta_{23}$;
if $y\kappa-z_{3}\zeta<0$, then $\theta_{13}<\theta_{23}<\theta_{12}$.
In particular, $\theta_{13}<0$.
\end{lem}
\begin{proof}
A direct calculation gives
\begin{align*}
\frac{y\zeta+z_{i}\kappa}{y\kappa-z_{i}\zeta}
-\frac{y\zeta+z_{j}\kappa}{y\kappa-z_{j}\zeta}
=\frac{y(z_{i}-z_{j})(\kappa^2+\zeta^2)}{(y\kappa-z_{i}\zeta)(y\kappa-z_{j}\zeta)}.
\end{align*}
Note that $y\kappa-z_{1}\zeta>y\kappa-z_{2}\zeta>0$.
If $y\kappa-z_{3}\zeta>0$,  then it is direct to obtain  
\[\frac{y\zeta+z_{1}\kappa}{y\kappa-z_{1}\zeta}
>\frac{y\zeta+z_{2}\kappa}{y\kappa-z_{2}\zeta}
>\frac{y\zeta+z_{3}\kappa}{y\kappa-z_{3}\zeta},\]
which leads to
$V_{12}>V_{13}>V_{23}$ by (\ref{cotV}),
and hence
$\theta_{12}<\theta_{13}<\theta_{23}$. 
Similarly, one has that if $y\kappa-z_{3}\zeta<0$, 
then $\theta_{13}<\theta_{23}<\theta_{12}$.


Next, we prove $\theta_{13}<0$.
If $y\zeta+z_{2}\kappa<0$, then one has from Lemma \ref{KVij2*} that $\theta_{13}<0$. 
If $y\zeta+z_{2}\kappa>0$, then one has from $y\kappa-z_{2}\zeta>0$ and (\ref{cotV}) that 
$\cot(yV_{13})>0$ and $yV_{13}\in(0,\frac{\pi}{2})$.
Note that
\begin{align*}
V_{0}-V^*_{13}
=\frac{y^{2}(m_r^{2}-m_l^{2})\left((m_l-z_{2})m_r+y^{2}+z_{2}m_l\right)}
{M_{13}m_r(y^{2}+m_l^{2})}.
\end{align*}
It follows from $m_r>m_l>z_{2}>0$, $M_{13}>0$ and Lemma \ref{cd1}
that $V>V_{0}>V^*_{13}$,
which leads to $\theta_{13}<0$ by Lemma \ref{KVij2*}.
The proof is completed.
\end{proof}

{
For signs of $\tau_{12}$, $\tau_{13}$ and $\tau_{23}$, we now establish the following results.
\begin{thm}\label{ThmM}
Assume $m_l<m_r$ and $\sigma_{10}=\overline{\sigma_{20}}=iy$ with $y>0$. 
One has $\tau_{13}<0$ for any $(\alpha_1,\alpha_2)\in\Delta$.
Furthermore, 
\begin{itemize}
\item[(i)]
if $V\ge V^*_{23}$, then $\tau_{23}<0$ for any $(\alpha_1,\alpha_2)\in\Delta$;

if $V<V^*_{23}$, then 
$\tau_{23}>0$ for $\alpha_1<\theta_{23}$ and $\alpha_2<P_{23}(\alpha_1)$;
$\tau_{23}<0$ for $\alpha_1>\theta_{23}$ or $\alpha_2>P_{23}(\alpha_1)$;

\item[(ii)]
if $y\kappa-z_{3}\zeta\ge0$, then $\tau_{12}<0$ for any $(\alpha_1,\alpha_2)\in\Delta$;

if $y\kappa-z_{3}\zeta<0$ and $M_{12}(V-V^*_{12})<0$, then $\tau_{12}>0$ for any $(\alpha_1,\alpha_2)\in\Delta$;

if $y\kappa-z_{3}\zeta<0$ and $M_{12}(V-V^*_{12})>0$, then 
$\tau_{12}<0$ for $\alpha_1<\theta_{12}$ and $\alpha_2<P_{12}(\alpha_1)$;
$\tau_{12}>0$ for $\alpha_1>\theta_{12}$ or $\alpha_2>P_{12}(\alpha_1)$.
\end{itemize}
\end{thm}
\begin{proof}
By (\ref{ykz23}) and Lemma \ref{V13V}, one has
$V(y\kappa-z_{2}\zeta)>0$ and $\theta_{13}<0$.
It follows from Theorem \ref{KVij2} that $\tau_{13}<0$ for any $(\alpha_1,\alpha_2)\in\Delta$.

Note that $V(y\kappa-z_{2}\zeta)>0$.
If $V\ge V^*_{23}$, then one has from Lemma \ref{KVij2*} that $\theta_{23}<0$.
It then follows from (ii) in Theorem \ref{KVij2} that $\tau_{23}<0$.
If $V<V^*_{23}$, then one has from Lemma \ref{KVij2*} that $\theta_{23}>0$.
It then follows from (ii) in Theorem \ref{KVij2} that 
$\tau_{23}>0$ for $\alpha_1<\theta_{23}$ and $\alpha_2<P_{23}(\alpha_1)$;
$\tau_{23}<0$ for $\alpha_1>\theta_{23}$ or $\alpha_2>P_{23}(\alpha_1)$.
The other cases can be obtained similarly.
\end{proof}

It is obvious that $\tau_{12}$ (resp. $\tau_{23}$) is continuous on $\Delta$
and the sign of $\tau_{12}$ (resp. $\tau_{23}$) changes at 
$(\alpha_1,P_{12}(\alpha_1))$ (resp. $(\alpha_1,P_{23}(\alpha_1))$).
Note that graphs of $P_{ij}$'s do not intersect (Proposition 3.10 in \cite{SL22}).

\begin{cor}
Assume $m_l<m_r$ and $\sigma_{10}=\overline{\sigma_{20}}=iy$ with $y>0$. 
If $V<V^*_{23}$, $y\kappa-z_{3}\zeta<0$ and $M_{12}(V-V^*_{12})>0$, then
\begin{itemize}
\item[(i)]
$\tau_{12}<\tau_{13}<0<\tau_{23}$ for $\alpha_1<\theta_{23}$ and $\alpha_2<P_{23}(\alpha_1)$;

\item[(ii)]
$\tau_{13}<\tau_{12}(\tau_{23})<0$ for
$P_{23}(\alpha_1)<\alpha_2<P_{12}(\alpha_1)$ or $\theta_{23}<\alpha_1<\alpha_2<P_{12}(\alpha_1)$;

\item[(iii)]
$\tau_{23}<\tau_{13}<0<\tau_{12}$ for $\alpha_1>\theta_{12}$ or $\alpha_2>P_{12}(\alpha_1)$.
\end{itemize}
\end{cor}
\begin{proof}
For $V<V^*_{23}$, $y\kappa-z_{3}\zeta<0$ and $M_{12}(V-V^*_{12})>0$,
one has from Lemma \ref{KVij2*} and Lemma \ref{V13V} that $\theta_{12}>\theta_{23}>0$.
Note that $\tau_{13}=\tau_{12}+\tau_{23}<0$.
The conclusion then follows from Theorem \ref{ThmM} directly.
\end{proof}}

\section{Conclusion}\label{summary} In this work, for ionic flow involving three ion species (two cations with different valences and one anion),  we continued the study in \cite{SL22} on effects of permanent charges on ionic fluxes. 

For the new case examined here, our results are consistent with some found in \cite{SL22}. For the case of pure imaginary eigenvalues, one always has $\tau_{13}<0$ though. 
Otherwise, the results in the new case study continue to support our conjectures raised in Section 5 of \cite{SL22}; that is, we believe that the answers to the following questions are negative:
\medskip

{\em Question 1. Can $\tau_{13}>0$ and $\tau_{23}>0$ occur simultaneously?}

{\em Question 2. Can $\tau_{12}>0$ and $\tau_{23}>0$ occur simultaneously?}

\bigskip
 It is a common view that, if an ionic mixture consists of only one cation and one anion ion species, the ionic mixture behaves close to the salt in the sense that the concentrations of the cation and anion are tightly related due to electroneutrality nearly everywhere. On the other hand, with an additional ion species (another cation in this work), there is an extra freedom --- the specifics (such as location characters $\alpha_1$ and $\alpha_2$) of permanent charges could coordinate two cation species using their boundary conditions to create behavors far more richer than what one may naively guessed. (This was indeed happened to the authors several times when \cite{SL22} and this paper were prepared.)  How much does this extra freedom allow one to do? The two case studies in \cite{SL22} and this paper have shown a number of generally unexpected phenomena. 
The two questions raised in \cite{SL22} and reiterated above concern  limitations of this extra freedom, just like the universality for the case of two ion species in (\ref{universal}).  We hope in the near future, answers to these questions could be available. We believe that our case studies, by no means trivial,  could stimulate further investigation on this problems and other extensions. 

 We remark that, if the two cation species have the same valence, then their flux ratios are the same. This result, as a special case of a general result, will appear in a forthcoming paper.

\bigskip


  \noindent
 {\bf Acknowledgment.} The authors thank the reviewers for their comments and suggestions that improve the manuscript. Ning Sun was partially supported by the Joint Ph.D. Training Program sponsored by the China Scholarship Council and a Graduate Innovation Fund of Jilin University.  
 Weishi Liu was partially supported by Simons Foundation Mathematics and Physical Sciences-Collaboration Grants for Mathematicians 581822.

\end{document}